\documentclass[a4paper,11pt,leqno]{amsart}
\usepackage{amsmath}
\usepackage{amsfonts}
\usepackage{amsthm}
\usepackage{mathrsfs}
\usepackage{amssymb}

\numberwithin{equation}{section}
\newcommand{\re}{\mathbb{R}}
\newcommand{\ze}{\mathbb{Z}}
\newcommand{\diam}{\mathrm{diam}}

\newtheorem{thm}{Theorem}[section]
\newtheorem{lem}[thm]{Lemma}

\theoremstyle{definition}

\newtheorem{ass}{Assumption}

\theoremstyle{remark}

\begin{document}

\title[Improving the Lieb-Robinson bound]{Improving the Lieb-Robinson bound \\ 
for long-range interactions}
\author[T. Matsuta]{Takuro Matsuta}
\address{Graduate School of Mathematical Sciences, 
University of Tokyo, 3-8-1 Komaba, Meguro Tokyo, 153-8914, JAPAN}
\email{matsuta@ms.u-tokyo.ac.jp}
\author[T. Koma]{Tohru Koma}
\address{Department of Physics, Gakushuin University, 
Mejiro, Toshima-ku, Tokyo 171-8588, JAPAN}
\email{tohru.koma@gakushuin.ac.jp}
\author[S. Nakamura]{Shu Nakamura}
\address{Graduate School of Mathematical Sciences, 
University of Tokyo, 3-8-1 Komaba, Meguro Tokyo, 153-8914, JAPAN}
\email{shu@ms.u-tokyo.ac.jp}

%%%%%%%%%%%%%%%%%%%%%%%%%%%%%%%%%%%%%%%%%%%
\begin{abstract}
We improve the Lieb-Robinson bound for a wide class of quantum many-body systems 
with long-range interactions decaying by power law. As an application, we show that 
the group velocity of information propagation grows by power law in time 
for such systems, whereas systems with short-range interactions exhibit 
a finite group velocity as shown by Lieb and Robinson.  
\end{abstract}

\maketitle

%%%%%%%%%%%%%%%%%%%%%%%%%%%%%%%%%%%%%%%%%%
\section{Introduction}

Lieb and Robinson \cite{LR} proved that the group velocity of information propagation  
is bounded by a finite constant in time for quantum many-body systems with short-range interactions 
(see also \cite{H,NS}). 
The Lieb-Robinson bound was extended to systems with long-range interactions decaying by 
power law \cite{1}. However, the resulting upper bound for the group velocity grows exponentially in time.   
If the upper bound is optimal and gives the true behavior of the group velocity in time, 
the information must spread to the space with such a fast-growing speed, which is unnatural physically \cite{EWMK}. 
Actually, the group velocity growing by power law in time was claimed by \cite{2} for 
a quantum spin system with two-body interactions decaying by power law.  
In the present paper, we extend, in a mathematically rigorous manner, their argument 
to a more general class of quantum many-body systems with long-range interactions decaying by power law. 
As a result, we prove that the group velocity of information propagation grows by power law in time. 

The present paper is organized as follows: In the next section, we give the precise definition of 
the models which we consider, and describe our main result, Theorem~\ref{thm:key}. 
The strategy for proving Theorem~\ref{thm:key} is developed in Sec.~\ref{sec:strategy}. 
The proof is given in Sec.~\ref{sec:ProofTheorem}.   
Appendices~\ref{LBfinite} and \ref{DerivIneq3rd} are devoted to technical estimates. 

%%%%%%%%%%%%%%%%%%%%%%%%%%%%%%%%%%%%%%%
\section{Models and main result}

Let $\Omega$ be a countable set with a metric $d(\cdot,\cdot)$, and we suppose this metric induces 
the discrete topology, i.e., each point $x\in\Omega$ is open and closed. We suppose there is a monotone 
increasing function $g(r)$ on $[0,\infty)$ and a constant $D>0$  such that 
\begin{equation}\label{eq-ass-vol}
\#\bigl\{y\in \Omega\big| d(x,y)\leq r\bigr\} \leq g(r)\leq C (1+r)^D, \quad r\geq 0, x\in\Omega
\end{equation}
with some $C>0$. We may consider $D$ as an analogue of the spatial dimension. 

We consider quantum spin systems on the point set $\Omega$. 
We assign a Hilbert space $\mathscr{H}_{x}$ to each site $x\in\Omega$.
Let $\Lambda$ be a finite subset of $\Omega$. Then, the configuration space of spin states on $\Lambda$ 
is given by the tensor product $\mathscr{H}_\Lambda=\bigotimes_{x\in\Lambda}\mathscr{H}_{x}$, 
and the algebra $\mathscr{A}_{\Lambda}=\bigotimes_{x\in\Lambda}\mathcal{B}(\mathscr{H}_x)$ 
of the observables on $\Lambda$ acts on 
the Hilbert space $\mathscr{H}_\Lambda$, where $\mathcal{B}(\mathscr{H}_x)$ denotes the 
Banach space of the bounded operators on $\mathscr{H}_x$. 
For $X\subset Y\subset\Omega$, we embed the algebra $\mathscr{A}_{X}$ on $X$ into 
$\mathscr{A}_{Y}$ on $Y$ by identifying $A\in\mathscr{A}_{X}$ with 
$A\otimes I\in\mathscr{A}_X\otimes\mathscr{A}_{(Y\setminus X)}\cong \mathscr{A}_{Y}$. 
The algebra of observables on $\Omega$ is defined as the completion of the local algebra 
$\mathscr{A}_{\text{loc}}=\bigcup\{ \mathscr{A}_{X}\,|\,
{X\subset\Omega, |X|<\infty}\}$ in the sense of the operator-norm topology. Here, $|X|$ stands for the number 
of the elements in the set $X$. 

Let $\Lambda$ be a finite subset of $\Omega$. Then, the Hamiltonian of a quantum spin system on 
$\Lambda$ is given by  
\begin{equation}
\label{HLambda}
H_{\Lambda}=\sum_{X\subset\Lambda}h_{X}, 
\end{equation}
where $h_{X}\in\mathscr{A}_{X}$ is the local Hamiltonian,\footnote{For an attempt 
for extending the Lieb-Robinson bounds to systems with an unbounded Hamiltonian, see, e.g., \cite{NSSSZ} 
and references therein.} i.e., a self-adjoint operator on $\mathscr{H}_X$, $X\subset\Omega$. 
The time evolution of the local observable $A\in\mathscr{A}_{\Lambda}$ by the generator $H_{\Lambda}$ 
is given by 
\begin{equation}
\tau_{t,\Lambda}(A)=e^{itH_{\Lambda}}Ae^{-itH_{\Lambda}},
\end{equation}
for the time $t\in\re$. 

We write $\diam(Z)$ for the diameter of the set $Z\subset\Omega$ which is given by 
$\diam(Z):=\max\{d(x,y)|\ x,y\in Z\}$. 
If $|Z|=+\infty$, then we define $\diam(Z)=+\infty$.  
Although we will consider general long-range interactions $h_X$ which include an arbitrary 
many-body interaction,  
we require the following assumption for the local Hamiltonian $h_X$: 

\begin{ass}
\label{ass:A}
\begin{enumerate}
\renewcommand{\theenumi}{(\roman{enumi})}
\item There is a decreasing function $f(R)$ on $[0,\infty)$  such that
\begin{equation}
\sup_{x\in\Omega}\sum_{\substack{Z\ni x; \\ \diam(Z)\ge R}}
\Vert h_{Z}\Vert\leq f(R), \quad R\geq 0. 
\label{eq:15}
\end{equation}
\item 
\begin{equation}\label{eq-LR-assumption}
\mathcal{C}_0 =\sup_{x\in\Omega} \, \sum_{y\in\Omega}\, \sum_{Z\ni x,y} \|h_Z\| <\infty. 
\end{equation}
\end{enumerate}
\end{ass}

A typical example is a spin system on $\Omega=\ze^D$, and we let $d(\cdot,\cdot)$ be the graph distance. 
We suppose it has only two-body interactions  $h_{\{x,y\}}$ for $x,y\in\ze^D$, 
and the following power-law decay condition with $\alpha>0$: 
\begin{equation}\label{eq-example-assumption}
\Vert h_X\Vert \leq \frac{\mathcal{C}_1}{[1+d(x,y)]^{\alpha+D}}, \ \ \ 
\text{for \ $X=\{x,y\}$},
\end{equation}
and $h_X=0$, otherwise, where $\mathcal{C}_1$ is some positive constant 
which is independent of the pair $\{x,y\}$ of the two sites. Then $\{h_X\}$ satisfies Assumption~A 
with $f(R)= C'(1+R)^{-\alpha}$. 
This is nothing but the case treated in \cite{2}

Our result is a mathematical justification of the argument in \cite{2}. 

\begin{thm}
\label{thm:key}
Let $A\in\mathscr{A}_X$ on $X\subset\Lambda$, and $B\in\mathscr{A}_Y$ on $Y\subset\Lambda$. 
Let $R\ge 1$, and write $r=d(X,Y)$. Then, 
\begin{align}
\nonumber 
\Vert[\tau_{t,\Lambda}(A),B]\Vert
&\leq 2\Vert A\Vert\; \Vert B\Vert\; |X|\; e^{vt-r/R} 
+ 4\Vert A\Vert\; \Vert B\Vert\; |X|\; t g(r)f(R) \\
&+2\mathcal{C}_2\Vert A\Vert\; \Vert B\Vert\; 
|X|^2\; tR(r\vee R)^Df(R)\; e^{vt-r/R}. 
\label{LBboundP}
\end{align}
for any $t\ge 0$, where $r\vee R:=\max\{r,R\}$, and $v$ and $\mathcal{C}_2$ 
are positive constants independent of $\Lambda$, $t$, $R$, $X$, $Y$, $A$ and $B$. 
\end{thm}

Let us explain the physical meaning of the resulting bound (\ref{LBboundP}) for the case of 
the hypercubic lattice $\ze^D$, 
with an additional assumption $\alpha>D$ and \eqref{eq-example-assumption}. 
We recall $g(r)=C(1+r)^D$ and $f(R)=C'(1+R)^{-\alpha}$. But, in this case, 
the factor $(r\vee R)^D$ in the third term in the right-hand side of (\ref{LBboundP}) 
can be replaced with $(r\vee R)^{D-1}$ 
by carefully calculating the bound in the proof of Lemma~\ref{lem:third-term} in Appendix~\ref{DerivIneq3rd}. 
See the remark at the end of Appendix B. Let $r\ge 1$ and $t>0$. 
We choose the parameter $R$ as 
\[
R=r^\kappa \quad \mbox{with \ } \kappa=\frac{D+1}{\alpha+1}.
\]
Substituting these into the right-hand side of the bound (\ref{LBboundP}), we obtain  
\begin{align*}
\Vert[\tau_{t,\Lambda}(A),B]\Vert
&\leq 2\Vert A\Vert\; \Vert B\Vert\; |X|\exp[vt-r^\eta]\\
&+\mathcal{C}_3\Vert A\Vert\; \Vert B\Vert\; |X|\; \frac{t}{r^\eta}
+\mathcal{C}_4 \Vert A\Vert\; \Vert B\Vert\; |X|^2\; \frac{t}{r^{2\eta}}\exp[vt-r^\eta],
\end{align*}
with $\eta=(\alpha-D)/(\alpha+1)$, where $\mathcal{C}_3$ and $\mathcal{C}_4$ are some positive constants, and 
we have used $\kappa<1$ which is derived from the assumption $\alpha>D$. 
We define the upper bound of the propagation distance by $r_{\rm max}(t)=(\lambda v t)^{1/\eta}$ 
as a function of time $t$ with the scale parameter $\lambda>1$. 
Substituting $r=r_{\rm max}(t)$ into the above upper bound, one has 
\begin{align*}
\Vert[\tau_{t,\Lambda}(A),B]\Vert
&\leq 2\Vert A\Vert\; \Vert B\Vert\; |X|\exp[-(\lambda-1)vt]
+\mathcal{C}_3\Vert A\Vert\; \Vert B\Vert\; |X|\; \frac{1}{\lambda v}\\
&+\mathcal{C}_4 \Vert A\Vert\; \Vert B\Vert\; |X|^2\; \frac{1}{(\lambda v)^2t}\exp[-(\lambda-1)vt]\\
&\sim \mathcal{C}_3\Vert A\Vert\; \Vert B\Vert\; |X|\; \frac{1}{\lambda v}
\end{align*}
for a large $t$. Clearly, for a large $\lambda$, the right-hand side becomes small, 
while it gives the order of $1$ for $\lambda$ of the order of $1$.  
Thus, the quantity $r_{\rm max}(t)$ gives the upper bound of the propagation distance as a function of time $t$.  
{From} the definition, it obeys the power law as  
\[
r_{\rm max}(t) = (\lambda v)^{1/\eta}t^{1+\gamma}\quad \mbox{for \ } t>0,
\]
with $\gamma=(D+1)/(\alpha-D)$.
The corresponding group velocity $v_{\rm g}(t)$ behaves as 
\[
v_{\rm g}(t):=\frac{d}{dt}r_{\rm max}(t)=(\lambda v)^{1/\eta}t^\gamma. 
\]
This exactly coincides with the behavior obtained in \cite{2}. For related recent numerical computations for 
quantum spin systems with long range interactions, see, e.g., \cite{HT,RGLSSFMGM}.     
For an overview of results and applications on the Lieb-Robinson bounds 
for quantum many-body systems, see, e.g., \cite{NS2} and references therein. 

%%%%%%%%%%%%%%%%%%%%%%%%%%%%%%%%%%%%%%%%%%%%%%%%%%%%%%%
\section{Decomposition of the Hamiltonian}
\label{sec:strategy}

In this section, we describe  our strategy for proving Theorem~\ref{thm:key}. 
The idea of decomposing the Hamiltonian was introduced in \cite{2}.  

Let $R$ be a positive number. We decompose the Hamiltonian $H_\Lambda$ of (\ref{HLambda}) 
into two parts as 
$$
H_{\Lambda}=H_{\Lambda}^{(<R)}+H_{\Lambda}^{(\geq R)}
$$
with
\begin{equation}
\label{HLambda<R}
H_{\Lambda}^{(<R)}=\sum_{Z\subset\varLambda}h_{Z}^{(<R)},
\end{equation}
and
\begin{equation}
\label{HLambda>R}
H_{\Lambda}^{(\geq R)}=\sum_{Z\subset\varLambda}h_{Z}^{(\geq R)},
\end{equation}
where the two local Hamiltonians, $h_{Z}^{(<R)}$ and $h_{Z}^{(\geq R)}$, are given by 
\[
h_{Z}^{(<R)}:=\begin{cases}
h_{Z}, \quad&  \text{if \ }\diam(Z)<R,\\
0, & \text{otherwise}, 
\end{cases}
\]
and $h_{Z}^{(\geq R)}:=h_Z-h_{Z}^{(< R)}$. The Hamiltonian $H_\Lambda^{(<R)}$ is the short-range 
part with the interaction range $R$, and $H_\Lambda^{(\ge R)}$ is the long-range part. 
Clearly, from Assumption~\ref{ass:A}, one has 
\begin{equation}
\sup_{x\in\Lambda}\sum_{Z\ni x}\bigl\Vert h_{Z}^{(\geq R)}\bigr\Vert\leq f(R)
\label{eq:sr-condition}
\end{equation}
for $R\ge 1$. 

The time evolution by the short-range Hamiltonian $H_\Lambda^{(<R)}$ is given by  
\[
\tau_{t,\Lambda}^{(<R)}(A)=e^{itH_{\Lambda}^{(<R)}}Ae^{-itH_{\Lambda}^{(<R)}}
\]
for a local observable $A \in\mathscr{A}_\Lambda$. We also introduce a unitary operator,  
\[
\mathscr{U}_\Lambda^R(t)=e^{itH_{\Lambda}^{(<R)}}e^{-itH_{\Lambda}},
\]
which satisfies the Schr\"odinger equation of the interaction picture, 
\begin{equation}
\label{eq:evo-eq}
i\frac{d}{dt}\mathscr{U}_\Lambda^R(t)=H_\Lambda^{(\geq R)}(t)\mathscr{U}_\Lambda^R(t),
\end{equation}
with the initial condition $\mathscr{U}_\Lambda^R(0)=1$, where we have written 
\begin{equation}
\label{HL>Rt}
H_\Lambda^{(\geq R)}(t):=\tau_{t,\Lambda}^{(<R)}\left(H_\Lambda^{(\geq R)}\right)
\end{equation}
for short. 
Then, as is well known, the time evolution of the observable $A\in\mathscr{A}_{\Lambda}$ 
by the total Hamiltonian $H_\Lambda$ of (\ref{HLambda}) 
is given by 
\[
\tau_{t,\Lambda}(A)=\left[\mathscr{U}_\Lambda^R(t)\right]^{*}\tau_{t,\Lambda}^{(<R)}(A)\mathscr{U}_\Lambda^R(t).
\]
{From}   
\begin{align*}
[\tau_{t,\Lambda}(A),B] &= [ \mathscr{U}_\Lambda^R(t)^*\tau_{t,\Lambda}^{(<R)}(A)\mathscr{U}_\Lambda^R(t), B] \\
&=\mathscr{U}_\Lambda^R(t)^*[\tau_{t,\Lambda}^{(<R)}(A), \mathscr{U}_\Lambda^R(t)B\mathscr{U}_\Lambda^R(t)^*]
\mathscr{U}_\Lambda^R(t),
\end{align*}
one has 
\begin{equation}
\left\Vert[\tau_{t,\Lambda}(A),B]\right\Vert
=\left\Vert[\tau_{t,\Lambda}^{(<R)}(A),\mathscr{U}_\Lambda^R(t)B\mathscr{U}_\Lambda^R(t)^{*}]\right\Vert 
\label{eq:com-identity1}
\end{equation}
for two local observables $A$ and $B$.

For $r>0$ and $X\subset \Lambda$, we define 
\[
\widetilde{X_{r}}:=\{x\in\varLambda\,|\,d(x,X)\leq r\},
\]
where $d(x,X):=\min\{d(x,y)|\ y\in X\}$. The set $\widetilde{X_r}$ is  the $r$-neighborhood of $X$.

\begin{lem} 
\label{lem:main-ineq}
Let $A\in\mathscr{A}_X$ on a finite subset $X$ of $\Omega$, and let $r>0$. Then,  
\begin{multline}
\left\Vert[\tau_{t,\Lambda}(A),B]\right\Vert
\leq\left\Vert[\tau_{t,\Lambda}^{(<R)}(A),B]\right\Vert
+2\Vert B\Vert
\sum_{Z\cap\widetilde{X_{r}}\neq\emptyset}
\int_{0}^{t}\bigl\Vert\bigl[\tau_{t-s,\Lambda}^{(<R)}(A),h_Z^{(\geq R)}\bigr]\bigr\Vert ds \\
+2\Vert B\Vert\sum_{Z\cap\widetilde{X_{r}}=\emptyset}
\int_{0}^{t}\bigl\Vert\bigl[\tau_{t-s,\Lambda}^{(<R)}(A),h_Z^{(\geq R)}\bigr]\bigr\Vert ds
\label{eq:main-ineq}
\end{multline}
for any $B\in\mathscr{A}_\Lambda$ and any $t\ge 0$. 
\end{lem}

\begin{proof}
We introduce a $\mathscr{A}_{\Lambda}$-valued function $f(t,s)$ for $s,t>0$ by 
\[
f(t,s)=[\tau_{t,\Lambda}^{(<R)}(A),\mathscr{U}_\Lambda^R(s)B\mathscr{U}_\Lambda^R(s)^{*}]. 
\]
Clearly, from \eqref{eq:com-identity1}, one has 
$$
\left\Vert[\tau_{t,\Lambda}(A),B]\right\Vert
= \Vert f(t,t)\Vert .
$$
By using the Schr\"odinger equation \eqref{eq:evo-eq} and the Jacobi identity, we have 
\begin{align*}
\frac{d}{ds}f(t,s) &= -i [\tau_{t,\Lambda}^{(<R)}(A), [H_\Lambda^{(\geq R)}(s), 
\mathscr{U}_\Lambda^R(s)B\mathscr{U}_\Lambda^R(s)^*]]\\
&=-i[H_\Lambda^{(\geq R)}(s),f(t,s)]
+i[\mathscr{U}_\Lambda^R(s)B\mathscr{U}_\Lambda^R(s)^{*},
[\tau_{t,\Lambda}^{(<R)}(A),H_\Lambda^{(\geq R)}(s)]].
\end{align*}
Here, we have used the notation of (\ref{HL>Rt}). 
Since the first term in the right-hand side in the second equality is 
a generator of a unitary evolution, we can apply a variation of the Duhamel 
principle  (see, e.g., Appendix of \cite{3}).
Therefore, by integrating both sides with respect $s$ from $0$ to $t$, we obtain 
\begin{align*}
\Vert f(t,t)\Vert &\le \Vert f(t,0)\Vert  +2\Vert B\Vert 
\int_0^t ds\bigl\Vert [\tau_{t,\Lambda}^{(<R)}(A),H_\Lambda^{(\ge R)}(s)]\bigr\Vert \\ 
&=\bigl\Vert[\tau_{t,\Lambda}^{(<R)}(A),B]\bigr\Vert
+2\Vert B\Vert 
\int_0^t ds\bigl\Vert [\tau_{t-s,\Lambda}^{(<R)}(A),H_\Lambda^{(\ge R)}]\bigr\Vert. 
\end{align*}
Substituting the expression (\ref{HLambda>R}) of $H_\Lambda^{(\ge R)}$ 
into the integrand  of the last integral, 
we obtain the desired bound (\ref{eq:main-ineq}). 
\end{proof}

%%%%%%%%%%%%%%%%%%%%%%%%%%%%%%%%%%%%%%%%%%%%%%%%%%%
\section{Proof of Theorem~\ref{thm:key}}
\label{sec:ProofTheorem}

Concerning the first term in the right-hand side of the inequality (\ref{eq:main-ineq}), 
we have 
\begin{equation}
\left\Vert[ \tau_{t,\Lambda}^{(<R)}(A),B]\right\Vert \leq 2\Vert A\Vert \Vert B\Vert |X|\;
\exp[vt-d(X,Y)/R]
\end{equation}
for $A\in\mathscr{A}_X$ and $B\in\mathscr{A}_Y$, from Theorem~\ref{thm:LR-bound} 
in Appendix~\ref{LBfinite}. 

 For the second term, by using the inequality \eqref{eq:sr-condition}, we obtain 
\begin{align*}
\sum_{Z\cap\widetilde{X_{r}}\neq\emptyset}\int_{0}^{t}\left\Vert\left[\tau_{t-s,\Lambda}^{(<R)}(A),
h_{Z}^{(\geq R)}\right]\right\Vert ds 
&\leq 2t \Vert A\Vert\sum_{Z\cap\widetilde{X_{r}}\neq\emptyset} \Vert h_Z^{(\geq R)}\Vert  \\
&\leq 2t\Vert A\Vert \sum_{x\in\widetilde{X_r}} \sum_{Z\ni x} \Vert h_Z^{(\geq R)}\Vert \\
&\leq 2t\Vert A\Vert |\widetilde{X_r}| f(R) \\
&\leq 2t\Vert A\Vert |X|  g(r) f(R).
\end{align*}

The third term is estimated  by Lemma~\ref{lem:third-term} in Appendix~\ref{DerivIneq3rd} as  
\begin{equation}
\label{3rdterm}
\sum_{Z\cap\widetilde{X_{r}}=\emptyset}
\int_{0}^{t}\bigl\Vert\bigl[\tau_{t-s,\Lambda}^{(<R)}(A),h_Z^{(\geq R)}\bigr]\bigr\Vert ds
\leq \mathcal{C}_2 t\Vert A\Vert |X|^2 (r\vee R)^D R f(R) e^{vt-r/R},
\end{equation}
where $\mathcal{C}_2$ is some positive constant. 

We set $r=d(X,Y)$. Combining these with Lemma~\ref{lem:main-ineq}, we have 
\begin{align*}
\Vert[\tau_{t,\Lambda}(A),B]\Vert
&\leq 2\Vert A\Vert \Vert B\Vert |X| e^{vt-r/R}
+ 4 t\Vert A\Vert\; \Vert B\Vert |X| g(r)f(R) \\
&+2\mathcal{C}_2 t\Vert A\Vert\; \Vert B\Vert |X|^2 (r\vee R)^D Rf(R) e^{vt-r/R}. 
\end{align*}
\qed 
%%%%%%%%%%%%%%%%%%%%
\appendix 

%%%%%%%%%%%%%%%%%%%%%%%%%%%%%%%%%%%%%%%%%%%%%%%%%%%%%%%%%%%%%%%
\section{The Lieb-Robinson bound for finite-range interactions}
\label{LBfinite}

In this appendix, we derive a Lieb-Robinson bound for the Hamiltonian $H_\Lambda^{(<R)}$ of (\ref{HLambda<R}) 
with finite-range interactions. 
The Lieb-Robinson bound for finite-range interactions is given by:  

\begin{thm}
\label{thm:LR-bound}
Let $A\in\mathscr{A}_{X}$
and $B\in\mathscr{A}_{Y}$ with $X,Y\subset\varLambda$, and let $R>0$. Then, we have  
\begin{equation}
\Vert [\tau_{t,\Lambda}^{(<R)}(A),B]\Vert 
\leq2\Vert A\Vert\Vert B\Vert|X|\exp[{vt-d(X,Y)/{R}}]
\label{eq:5}
\end{equation}
for any $t\geq 0$ with some positive constant $v$, under Assumption~\ref{ass:A}-\rm{(ii)}.
\end{thm}

\begin{proof}
We essentially follow the proof of the Lieb-Robinson bound in \cite{4}, with explicit control of the 
constants. We set
\[
C_B(Z,t)=\sup_{A\in\mathscr{A}_{Z}}\frac{\Vert[\tau_{t,\Lambda}^{(<R)}(A),B]\Vert}{\Vert A\Vert},
\quad \text{for }Z\subset \Lambda.
\]
By computations similar to the proof of Lemma~\ref{lem:main-ineq}, we learn 
\begin{align*}
\frac{d}{dt}[\tau_{t,\Lambda}^{(<R)}(A),B]
&=i\sum_{Z\cap X\neq \emptyset}[\tau_{t,\Lambda}^{(<R)}(h_Z^{(<R)}),
[\tau_{t,\Lambda}^{(<R)}(A),B]]\\ 
&\quad -i\sum_{Z\cap X\neq \emptyset}[\tau_{t,\varLambda}^{(<R)}(A),
[\tau_{t,\Lambda}^{(<R)}(h_Z^{(<R)}),B]]. 
\end{align*}
Since the first term in the right-hand side is a generator of norm-preserving evolution, we have 
\begin{equation*}
\Vert[\tau_{t,\Lambda}^{(<R)}(A),B]\Vert\leq\Vert[A,B]\Vert+2\Vert A\Vert\sum_{Z\cap X\neq \emptyset}
\int_{0}^{t}\Vert[\tau_{s,\Lambda}^{(<R)}(h_Z^{(<R)}),B]\Vert ds.
\end{equation*}
in the same way as in Lemma~\ref{lem:main-ineq}.
Consequently, we obtain 
\[
C_{B}(X,t)\leq C_{B}(X,0)+2\sum_{Z\cap X\neq \emptyset}\bigl\Vert h_Z^{(<R)}\bigr\Vert\int_{0}^{t}C_{B}(Z,s)ds.
\]
 Iterations of this inequality yield 
\begin{equation}
\label{eq:CB-estimate}
C_{B}(X,t)\leq C_B(X,0)+2\Vert B\Vert\sum_{n=1}^{\infty}\frac{(2t)^{n}}{n!}a_{n}, 
\end{equation}
where 
\[
a_{n}=\sum_{Z_1\cap X\neq \emptyset}
\sum_{Z_{2}\cap Z_1\neq \emptyset}\cdots\sum_{\substack{Z_{n}\cap Z_{n-1}\neq\emptyset\\
Z_{n}\cap Y\neq\emptyset}
}\prod_{i=1}^{n}\Vert h_{Z_{i}}^{(<R)}\Vert.
\]
Using Assumption~(A)-(ii), we have  
\[
a_{1} \leq \sum_{x\in X}\sum_{y\in Y}\sum_{Z\ni x,y}\Vert h_Z^{(<R)}\Vert
\leq  \sum_{x\in X} \mathcal{C}_0 \leq \mathcal{C}_{0}|X|.
\]
Similarly, we have 
\begin{align*}
a_{2} & \leq  \sum_{x\in X}\sum_{y\in Y}\sum_{z\in\Lambda}\sum_{Z_{1}\ni x,z}\sum_{Z_{2}\ni z,y}
\bigl\Vert h_{Z_{1}}^{(<R)}\bigr\Vert\; \bigl\Vert h_{Z_{2}}^{(<R)}\bigr\Vert
\leq  \mathcal{C}_{0}^{2}|X|. 
\end{align*}
Repeating this procedure, we obtain 
\begin{equation}
a_{n}\leq\mathcal{C}_{0}^{n}|X|, \quad \text{for }n\geq 1. 
\label{eq:8}
\end{equation}

On the other hand, we have 
\begin{equation}
a_{n}=0\quad\text{if \ }nR<d(X,Y)
\label{eq:9}
\end{equation}
because $h_Z^{(<R)}=0$ for $\diam(Z)\geq R$. 
Combining \eqref{eq:CB-estimate}, (\ref{eq:8}) and (\ref{eq:9}), we have 
\begin{align*}
C_B(X,t) & \leq  2\Vert B\Vert|X|\sum_{n\geq d(X,Y)/{R}}\frac{(2\mathcal{C}_{0}t)^{n}}{n!}\\
&\leq 2\Vert B\Vert|X|\sum_{n\geq d(X,Y)/{R}}\frac{(2\mathcal{C}_{0}t)^{n}}{n!}e^{n-d(X,Y)/{R}}\\
&\leq 2\Vert B\Vert|X|\sum_n\frac{(2e\mathcal{C}_{0}t)^{n}}{n!}e^{-d(X,Y)/{R}}
%\\&
= 2\Vert B\Vert|X|e^{vt-d(X,Y)/{R}} 
\end{align*}
for $d(X,Y)>0$. Here, the group velocity $v$ is given by $v=2e\mathcal{C}_{0}$.
This completes the proof.
\end{proof}

%%%%%%%%%%%%%%%%%%%%%%%%%%%%%%%%%%%%%%%%%%%%%%%
\section{Derivation of the inequality (\ref{3rdterm})}
\label{DerivIneq3rd}

The third term in the right-hand side of (\ref{eq:main-ineq}) in Lemma~\ref{lem:main-ineq} 
is estimated as  follows.

\begin{lem}
\label{lem:third-term}
Let $A\in\mathscr{A}_X$, $X\subset\Lambda$. Then 
\[
\sum_{Z\cap\widetilde{X_r}=\emptyset} \Vert [\tau_{t,\Lambda}^{(<R)}(A),h_Z^{(\geq R)}]\Vert 
\leq \mathcal{C}_2 \Vert A\Vert |X|^2 (r\vee R)^D Rf(R) e^{vt-r/R}
\]
for $t\geq 0$, $r>0$ and $R\geq 1$, where $r\vee R:=\max\{r,R\}$, and $\mathcal{C}_2$ is some positive constant. 
\end{lem}

\begin{proof}
{From} Theorem~\ref{thm:LR-bound}, we have 
\begin{align*}
&\sum_{Z\cap\widetilde{X_r}=\emptyset} \Vert [\tau_{t,\Lambda}^{(<R)}(A),h_Z^{(\geq R)}]\Vert 
\leq 2 \Vert A\Vert\; |X|\sum_{Z\cap\widetilde{X_r}=\emptyset} 
\Vert h_Z^{(\geq R)}\Vert e^{vt-d(X,Z)/R} \\
&\qquad = 2\Vert A\Vert\; |X| \sum_{k=0}^\infty \sum_{\substack{Z:\\ r+k<d(X,Z)\leq r+k+1}} 
\Vert h_Z^{(\geq R)}\Vert 
e^{vt-(r+k)/R}\\
&\qquad \leq 2\Vert A\Vert\; |X| \sum_{k=0}^\infty \sum_{\substack{z: \\d(X,z)\leq r+k+1}} 
\sum_{Z\ni z}\Vert h_Z^{(\geq R)}\Vert 
e^{vt-(r+k)/R}\\
&\qquad \leq 2\Vert A\Vert\; |X|  f(R) e^{vt}  \sum_{x\in X} \sum_{k=0}^\infty 
\sum_{\substack{z:\\ d(x,z)\leq r+k+1}} 
e^{-(r+k)/R}, \\
&\qquad \leq 2\Vert A\Vert\; |X|^2  f(R) e^{vt} 
\sup_x \sum_{k=0}^\infty \sum_{\substack{z:\\ d(x,z)\leq r+k+1}}e^{-(r+k)/R},
\end{align*}
where we have used the inequality (\ref{eq:sr-condition})  to show the third inequality. 
Elementary computations yield
\begin{align}
\sum_{k=0}^\infty \sum_{\substack{z:\\ d(x,z)\leq r+k+1}} e^{-(r+k)/R}
&\leq C_1 \int_r^\infty g(y+1) e^{-y/R}dy \nonumber \\
&\leq C_2  R^{D+1} \int_{r/R}^\infty y^D e^{-y}dy \nonumber \\
&\leq C_3 R^{D+1} (r/R+1)^D e^{-r/R} \nonumber \\
&\leq C_3 (r\vee R)^DR e^{-r/R}
\label{sumexpbound}
\end{align}
for each $x\in X$, where the constants $C_1,C_2,C_3$ depend only on the constants in \eqref{eq-ass-vol}. 
Combining these, we conclude the assertion. 
\end{proof}

Finally, we remark on the typical example in Section~2: The first inequality in (B.1) 
can be replaced by a sum in $z$ : $r+k<d(x,z)\leq r+k+1$ in general, 
and hence the factor $g(y+1)$ is replaced by $C y^{D-1}$ in the case of 
$\mathbb{Z}^D$ lattice. Therefore, the factor $(r\vee R)^D$ in 
the right-hand side of the fourth inequality can be replaced with $(r\vee R)^{D-1}$.

%%%%%%%%%%%%%%%%%%%%%%%%%%%%%%%%%%%%%%%%%%%%%%%%%%%%%%%%%%%%%%%%%%%%%%%%%

\end{document}